\documentclass[letterpaper, 10 pt, conference]{ieeeconf}
\IEEEoverridecommandlockouts                            
\usepackage{subcaption}
\usepackage{graphicx}
\usepackage{amsmath}
\usepackage{amssymb}

\newtheorem{theorem}{Theorem}
\newtheorem{lemma}{Lemma}

\renewcommand{\P}{\mathbb{P}}
\newcommand{\E}{\mathbb{E}}
\newcommand{\Poisson}{\mathrm{Poisson}}
\newcommand{\Rbb}{\mathbb{R}}
\newcommand{\Zbb}{\mathbb{Z}}

\newcommand{\goesto}{\rightarrow}
\newcommand{\set}[1]{\left\{#1\right\}}

\usepackage{flushend}

\title{\LARGE \bf Infinite Server Queueing Networks with Deadline
  Based Routing }

\author{Neal Master and Nicholas Bambos
  \thanks{Neal Master is supported a Stanford Graduate Fellowship
    (SGF) in Science \& Engineering.}
  \thanks{N. Master and N. Bambos are with the Department of
    Electrical Engineering, Stanford University, Stanford, CA, 94305,
    USA.  {\tt\small\{nmaster, bambos\}@stanford.edu}} }

\begin{document}

\maketitle
\thispagestyle{empty}
\pagestyle{empty}

\begin{abstract}
  Motivated by timeouts in Internet services, we consider networks of
  infinite server queues in which routing decisions are based on
  deadlines. Specifically, at each node in the network, the total
  service time equals the minimum of several independent service times
  (e.g. the minimum of the amount of time required to complete a
  transaction and a deadline). Furthermore, routing decisions depend
  on which of the independent service times achieves the minimum
  (e.g. exceeding a deadline will require the customer to be routed so
  they can re-attempt the transaction). Because current routing
  decisions are dependent on past service times, much of the existing
  theory on product-form queueing networks does not apply. In spite of
  this, we are able to show that such networks have product-form
  equilibrium distributions. We verify our analytic characterization
  with a simulation of a simple network. We also discuss extensions of
  this work to more general settings.
\end{abstract}

\section{Introduction}
Infinite server queues are an important stochastic modeling tool for a
diverse array of disciplines: they are used to study Internet
services~\cite{Urgaonkar_MultiTier_2005}, fragmentation in dynamic
storage allocation~\cite{Coffman_Fragmentation_1985}, software
reliability~\cite{Dohi_SW_2002}, and call
centers~\cite{Whitt_CallCenters_1999}. Although having an infinite
number of servers is typically only an approximation of reality, such
approximations are often very useful. In the case of Internet
services, these approximations are also becoming increasingly
reasonable. Cloud based web hosting from companies like Amazon,
Google, and Microsoft allows businesses to dynamically provision their
web services so they can accommodate almost arbitrarily large customer
demand \cite{Zhang_Cloud_2010}. Because customer demand is (almost)
always met immediately, infinite server queues offer both natural and
tractable stochastic models.

While infinite server queues are useful by themselves, they become
even more powerful when used in conjunction with queueing network
theory. In particular, for many types of queueing networks it is easy
to compute product-form stationary distributions for the number
customers at each node in the network. The term product-form refers to
the fact that the joint distribution is merely a product of the
marginals and this drastically simplifies the analysis of seemingly
intractable systems. The first major theoretical result on
product-form distributions for queueing networks was presented by
Jackson \cite{Jackson_1963}. This seminal work was later extended by
Baskett, Chandy, Muntz, and Palacios in the study of what are now
called BCMP networks \cite{BCMP_1975}. Subsequently, Kelly used the
idea of quasi-reversibility to study Jackson and BCMP networks in the
more general framework of Kelly networks
\cite{Kelly_Networks_1979}. Kelly networks have been used and built
upon in several ways; see \cite[Chapter~4]{Chen_Yao} and the
references therein for more details. 

A particularly interesting development is the notion of insensitivity
to the service time distributions (see \cite[Chapter~3]{Walrand} and
the references therein). Insensitivity refers to the fact that many
product-form distributions depend on the service time distributions
only through their means and hence, general service times can easily
be incorporated into queueing network models. For instance, even if
the service times are not exponential, applying Jackson's Theorem
\cite{Jackson_1963} as if they were exponential will often yield the
correct product-form distribution. Insensitivity results still
typically require Poisson arrivals, but allowing for general service
time distributions offers considerable modeling flexibility.

Although the product-form results mentioned above are very general and
powerful, they typically require Markovian routing. Formally, this
means that when a customer completes service at a particular node and
is routed elsewhere, the routing decision must be independent of the
past evolution of the network.  In particular, requiring Markovian
routing precludes routing that is dependent on previous service
times. This is limiting because we are interested in a model in which
routing decision are based on service times through service time
constraints, e.g. deadlines.

In particular, we are motivated by Internet services with timeouts
\cite{Russo_Timeouts_2009}. For example, consider a customer who uses
an online financial service with multi-factor authentication for
logging in. There is a deadline for the verification process -- if the
customer takes too long then the verification process will fail and
the customer will need to try again. In this scenario, we can model
the log-in step as an infinite server queue because (assuming a cloud
hosting service is used) an arbitrary number of customers can sign-in
without waiting. The total service time is given by the minimum of two
independent service times -- the amount of time required to sign-in
and a deadline. Moreover, the routing decision depends on which
service time achieves this minimum. Indeed, if the timeout occurs and
the service time is the deadline, then the customer will be re-routed
back to the sign-in page. If the timeout does not occur then the
customer will successfully enter the system. Even in this very simple
scenario, we see that timeouts and deadlines can drive routing
decisions. As a result, the routing is generally not Markovian and we
are unable to apply the typical queueing theory to find the
equilibrium distribution of the system.

Although computing the equilibrium distribution of a queueing system
with non-Markovian routing is not necessarily straightforward, knowing
the equilibrium distribution can be quite useful. In the case of
Internet services, the equilibrium distribution can be used for
marketing. For instance, the average number of customers on a page and
the average amount of time they spend on that same page are useful for
pricing display advertisements \cite{Goldfarb_Advertising_2011}. While
this information can be collected empirically, getting the full
distribution can be time consuming if the services times are
heavy-tailed, as is often the case
\cite{Downey_LongTailed_2001}. Consequently, there is substantial
value in having an analytic model that can be used to understand and
design web services without running full experiments.

Given this motivation and background, the remainder of the paper is
organized as follows. In Section~\ref{sec:model_theory}, we present
our model and position it relative to the existing theory. Since our
model cannot be immediately solved by applying existing results, in
Section~\ref{sec:result} we show that our model does in fact have a
product-form equilibrium distribution. In
Section~\ref{sec:simulation}, we verify this analytic result with a
simulation. We conclude in Section~\ref{sec:conclusion}.

\section{Model Formulation and\\Limitations of Existing Network Theory\label{sec:model_theory}}
In this section, we formally outline the stochastic model of
interest. We then discuss the applicability of the theory of Jackson
and Kelly networks. In particular, we show that if all service times
are exponentially distributed then we can apply Jackson's Theorem
\cite{Jackson_1963} to find the product-form distribution. However, if
the service times are general, then we cannot apply insensitivity and
quasi-reversibility results (e.g. \cite[Chapter~3]{Walrand}) to find a
product form equilibrium distribution. This demonstrates the nuances
of our model.

\subsection{Model Formulation and Motivation\label{sec:model}}
We want to model infinite capacity service systems in which customers'
service times and routing probabilities are impacted by deadlines. For
example, consider a system in which customers experience a natural
service time but are ejected if their service time exceeds a fixed
deadline. In this case, the total service time is the minimum of the
natural service time and the deadline. Moreover, the customer will be
routed differently based on whether or not he experienced the deadline
or the natural service time. More generally, we consider systems in
which customers experience a service time that is the minimum of
several independent service times. Customers are routed stochastically
but the routing matrix will change based on which of the competing
service times achieved the minimum.

Formally, we consider a queueing network with $N$ nodes. Each
node has an infinite number of servers. At node $n$, customers
experience $K$ independent competing service times
$$\set{S_{n,1}, \hdots, S_{n,K}}$$
with a total service time of 
$$S_n = \min\set{S_{n,1}, \hdots, S_{n,K}}.$$
We assume that the probability that more than one of the service times
achieves the minimum is zero. We allow for some of these competing
service times to be infinite. Let $S_{n,k}^{(i)}$ denote the $k^{th}$
service time for customer $i$ at node $n$ and
$$S_n^{(i)} = \min\set{S_{n,1}^{(i)}, \hdots, S_{n,K}^{(i)}}.$$
We assume that each customer's service times are independent of all
other customers' service times. Furthermore, if a customer is served
by the same node multiple times, then these service times are also
independent. If 
$$S_{n,k}^{(i)} = \min\set{S_{n,1}^{(i)}, \hdots,S_{n,K}^{(i)}} = S_n^{(i)}$$ 
then customer $i$ is routed from node $n$ to node $m$ with probability
$P^k_{n,m}$ and the customer exits the network with probability $1 -
\sum_{m=1}^N P^k_{n,m}$. Customers arrive to node $n$ according to a
Poisson process with rate $\lambda_n$. We assume that every customer
spends only a finite amount of time in the network (i.e. the network
is open).

We note that because we are interested in infinite server nodes, we
can focus on a single class of customers without any loss of
generality. This is because the jobs do not interact in a queueing
buffer. Indeed, if we had multiple classes, we could consider a
``copy'' of the network for each class. As long as customers do not
change classes, these copies will not interact just as the different
customer classes do not interact. The arrival rates would differ for
each class but the model would not fundamentally change.

\subsection{Limitations of Jackson and Kelly Network Theory\label{sec:jackson}}
Jackson's Theorem \cite{Jackson_1963} is a celebrated result for
attaining product-form equilibrium distributions of queueing
networks. Although Jackson's Theorem applies only to queueing networks
with exponential service times, the results can often be extended to
the case of general service times (e.g. Kelly
networks~\cite{Kelly_Networks_1979}) with insensitivity arguments,
e.g.~\cite[Chapter~3]{Walrand}. We first show how Jackson's Theorem
can be applied to our model in the case that the service times are
exponential. We then show that unlike many other models, the
equilibrium distribution depends on the service time distributions in
their entirety, not merely through their means and hence,
insensitivity arguments do not apply to our model.

Theorem~\ref{thrm:jackson} is the version of Jackson's Theorem
presented in \cite[Chapter~2]{Chen_Yao}. Jackson's Theorem was
originally presented in \cite{Jackson_1963}.

\begin{theorem}[Jackson's Theorem]

  Suppose we have an $N$ node queueing network in which jobs arrive at
  rate $\lambda$. Arriving jobs are independently routed to node $j$
  with probability $p_{0j}$ where $\sum_{j=1}^N p_{0j} = 1$. Upon
  service completion at node $i$, a job is routed to node $j$ with
  probability $P_{ij}$ and leaves the network with probability $p_{i0}
  = 1 - \sum_{j=1}^N P_{ij}$. We assume that $P$ is substochastic
  (i.e. at least one of the row sums is strictly less than one so the
  network is open).

  When there are $x_i$ jobs at node $i$, the exponential service time
  has rate $\mu_i(x_i)$ where $\mu_i : \Zbb_+ \goesto \Rbb_+$ with
  $\mu_i(0) = 0$ and $\mu_i(x) > 0$ for all $x > 0$.

  Let $p_0 \in \Rbb^N$ where the $i^{th}$ entry is $p_{0i}$ and define
  $\alpha$ as the solution\footnote{A unique solution exists because
    $P$ is substochastic. Furthermore, this solution is
    non-negative. See \cite[Chapter~2]{Chen_Yao} for details.}  to the
  following equation:
  $$\alpha = \lambda p_0 + P^T \alpha$$
  Let $M_i(n)$ be defined as 
  $$M_i(n) = \prod_{j=1}^n \mu_i(j)$$
  and assume the following:
  $$\sum_{n=1}^\infty \frac{\alpha_i^n}{M_i(n)} < \infty$$

  Let $X_i$ be the number of jobs at node $i$ in equilibrium. Then
  $$\P(X_1 = x_1, \hdots, X_N = x_N) = \prod_{i=1}^N \P(X_i = x_i)$$
  where
  $$\P(X_i = x) = \P(X_i = 0) \alpha_i^n / M_i(n)$$ and 
  $$\P(X_i = 0) = \left(1 + \sum_{n=1}^\infty\frac{\alpha_i^n}{M_i(n)}\right)^{-1}.$$

\label{thrm:jackson}
\end{theorem}

To apply Theorem~\ref{thrm:jackson}, we need to focus on the case when
$S_{n,k}$ is exponentially distributed with rate $\mu_{n,k}$. Now we
will need to make use of the following elementary result:

\begin{lemma}[{\cite[Fact~2.3.1]{Walrand}}]
  \label{lem:exp}
  Let $\set{\sigma_k,\, 1 \leq k \leq K}$ be independent exponential
  random variables with respective rates $\lambda_k$. Then for all $t
  \geq 0$ and $1 \leq i \leq K$, 
  $$\P(\sigma_i = \min\set{\sigma_k,\, 1 \leq k \leq K} > t) = \frac{\lambda_i}{\lambda}e^{-\lambda t}$$
  with $\lambda = \lambda_1 + \cdots + \lambda_K$.
\end{lemma}

To quote \cite{Walrand}, this lemma ``states that the index of the
smallest of $K$ independent exponential random variables is
independent of the value of that minimum which is also exponentially
distributed.''  This tells us that a customer at node $n$ experiences
an exponential service time with rate $\mu_n = \sum_{k=1}^K \mu_{n,k}$
and that after service is complete the customer is routed to station
$m$ with probability 
$$Q_{n,m} = \sum_{k=1}^K \frac{\mu_{n,k}}{\mu_n} P^k_{n,m}$$
independently of the service time. In addition, we know that the
service rate function at node $n$ is $x \mapsto \mu_n x$. If $X_n$ is
the number of customers at node $n$ in equilibrium, then applying
Theorem~\ref{thrm:jackson} gives us that
$$X_n \sim \Poisson(\rho_n)$$
where $\rho_n = \alpha_n / \mu_n$ and $\alpha$ solves
$$\alpha = \lambda + Q^T \alpha.$$
Furthermore, Theorem~\ref{thrm:jackson} tells us that
$$\P(X_1 = x_1, \hdots, X_N = x_N) = \prod_{n=1}^N \P(X_n = x_n),$$
i.e. we have a product-form distribution.

This seems to be a very powerful result. Not only do we have a
product-form distribution, the distribution depends on the service
times only through their first moments. It is tempting to assume an
insensitivity result such as \cite[Theorem~3.3.2]{Walrand}:

\begin{theorem}
  \label{thrm:walrand}
  Suppose the $N$ nodes in Theorem~\ref{thrm:jackson} are infinite
  server queues with general and independent service times. Assume
  $\mu_i^{-1}$ is the mean service time at node $i$. We continue to
  assume that the arrivals to the network are Poisson and that the
  routing decisions are stochastic and independent of the past
  evolution of the network. Then
  $$\P(X_1 = x_1, \hdots, X_N = x_N) = \prod_{i=1}^N \P(X_i = x_i)$$
  where
  $$X_i \sim \Poisson(\alpha_i/\mu_i)$$
  and $\alpha$ solves the same linear flow equations from
  Theorem~\ref{thrm:jackson}.
\end{theorem}

Informally, this means that we can apply Theorem~\ref{thrm:jackson}
even when the service times are non-exponential. This kind of result
applies for queues besides the infinite server queue; the key
requirement is that the queue be quasi-reversible (see
\cite[Chapter~3]{Walrand} and the references therein for a discussion
of quasi-reversible queues and insensitivity results).

For our model, when considering exponential services times, the
probability that a customer at node $n$ will go to node $m$ after
service is $Q_{n,m}$. However, for general service times the
probability is
$$\tilde Q_{n,m} = \sum_{k=1}^K \P(S_{n,k} = \min\set{S_{n,1}, \hdots, S_{n, K}}) P^k_{n,m}$$
and in general $\tilde Q_{n,m} \ne Q_{n,m}$. Consequently, assuming
full insensitivity does not seem correct. Another approach would be to
use $\tilde Q$ instead of $Q$ but continue to use $\mu_n$ as the
service rate at node $n$. Indeed, the term ``insensitivity'' typically
refers to insensitivity of the service time distributions so it makes
sense that the routing matrix should change. Although less na\"ive
than assuming full insensitivity, we will see that this approach is
also incorrect.

Both na\"ive insensitivity approaches are incorrect for two
reasons. First note that if the service distributions change, then
$\mu_n^{-1}$ will generally not be the mean service time at node
$n$. However, even if the mean service times were preserved, applying
Theorem~\ref{thrm:walrand} would still be incorrect because the
routing is not Markovian. As mentioned earlier, Markovian routing is a
form of probabilistic routing for which the routing decisions do not
depend on the past evolution of the network. Markovian routing is
required for both Theorem~\ref{thrm:jackson} and
Theorem~\ref{thrm:walrand}. In our model, the routing decisions are
based on the service time that achieves the minimum and consequently
the routing decisions are not independent of the past evolution of the
network.

We see that Lemma~\ref{lem:exp} and Theorem~\ref{thrm:jackson} provide
us with a product-form result when the service times are exponential,
but these same arguments do not extend to the case of non-exponential
service times. We will demonstrate by simulation in
Section~\ref{sec:simulation} that not only are the two na\"ive
approaches (assuming full insensitivity and assuming insensitivity of
the service times) not mathematically justified, they also give
incorrect forms for the equilibrium distribution.

\section{A Product-Form Equilibrium Distribution\label{sec:result}}
The discussion in Section~\ref{sec:jackson} explains why typical
mathematical arguments cannot be applied to find a product-form
distribution for the queueing model described in
Section~\ref{sec:model}. However, in this section we are able to
explicitly characterize the stationary distribution as a
product-form. The key insight is to construct a queueing network for
which the routing is Markovian and also has a stochastically
equivalent equilibrium distribution. We are able to do this because
infinite server queueing nodes can be represented as several infinite
server queueing nodes acting in parallel. Although it may seem that
adding more nodes adds more servers, because we are dealing with
infinite server queues to begin with, the total number of servers is
actually preserved.

\begin{theorem}
  \label{thrm:result}
  Consider the queueing model from Section~\ref{sec:model} and let
  $X_n$ be the number of customers at node $n$ in equilibrium. For $n
  \in \set{1, \hdots, N}$ and $k \in \set{1, \hdots, K}$, define
  $$\mu_{n,k} = \left(\E[S_{n,k} | S_{n,k} = S_n]\right)^{-1}$$
  and let $\alpha \in \Rbb^{N \times K}$ be defined by the following
  equations:
  \begin{equation}
    \label{eq:flow}
    \alpha_{n,k} %
    = \P(S_{n,k} = S_n)\lambda_n %
    + \sum_{m,\ell} \P(S_{n,k} = S_n)P^\ell_{m,n}\alpha_{m,\ell}%
  \end{equation}
  Then
  $$X_n \sim \Poisson(\rho_n)$$
  where $\rho_n = \sum_{k=1}^K \alpha_{n,k} / \mu_{n,k}$.
  Furthermore,
  $$\P(X_1 = x_1, \hdots, X_N = x_N) = \prod_{n=1}^N \P(X_n = x_n),$$
  i.e. we have a product-form distribution.
\end{theorem}

\begin{proof}
  Consider the following $N \times K$ node network of infinite server
  queues. At node $(n, k)$, the service time is distributed as
  $$S_{n,k} | (S_{n,k} = S_n)$$
  so the service rate at node $(n, k)$ is $\mu_{n,k}$. Upon completing
  service at node $(n, k)$ a customer is routed in a Markovian fashion
  to node $(m, \ell)$ with probability
  $$P^k_{n,m}\P(S_{m,\ell} = S_m).$$
  With probability
  $$1 - \sum_{m,\ell} P^k_{n,m}\P(S_{m,\ell} = S_m)$$
  a customer leaves the network after service at node $(n,
  k)$. Customers arrive externally to node $(n, k)$ according to a
  Poisson process with rate $\P(S_{n,k} = S_n)\lambda_n$.

  Because the arrivals are Poisson and
  Poisson-Arrivals-See-Time-Averages (PASTA~\cite{Wolff_PASTA_1982}),
  we can show that the equilibrium distribution of this network is
  stochastically equivalent to the equilibrium distribution of the
  network in Section~\ref{sec:model} by considering the perspective of
  customers who arrive to various nodes. We will refer to nodes
  $(n,1), \hdots, (n,K)$ as supernode $n$ and we will show that in
  equilibrium supernode $n$ is stochastically equivalent to node $n$
  in the original network. Furthermore, we will show that customers
  are routed between supernodes in a manner that is stochastically
  equivalent to the manner in which customers are routed between nodes
  in the original network.

  \begin{enumerate}
  \item{\it Service Times:}
    Among all (internal and external) arrivals to supernode $n$, the
    service time distribution is
    $$\sum_{k=1}^K \P(S_{n,k} = S_n) \times S_{n,k} | (S_{n,k} = S_n) = S_n.$$
    In addition, in the original network a customer will experience
    service time $S_{n,k} | (S_{n,k} = S_n)$ with probability
    $\P(S_{n,k} = S_n)$. Therefore, we have that the service times in
    supernode $n$ are stochastically equivalent to the services times at
    node $n$ in the original network.

  \item{\it External Arrivals:} The total external arrival rate to
    supernode $n$ is
    $$\sum_{k=1}^K\P(S_{n,k} = S_n)\lambda_n = \lambda_n\sum_{k=1}^K\P(S_{n,k} = S_n) = \lambda_n$$
    which is the external arrival rate to node $n$ in the original
    network. 

  \item{\it Routing:} Now consider how customers are routed from
    supernode $n$ to supernode $m$. First note that since a customer
    at node $(n, k)$ is routed to node $(m, \ell)$ with probability
    $P^k_{n,m} \P(S_{m,\ell} = S_m)$, the probability of being routed
    from supernode $n$ to supernode $m$ is
    \begin{align*}
      &\sum_k \sum_\ell P^k_{n,m} \P(S_{m,\ell} = S_m) \\
      &= \sum_k P^k_{n,m} \sum_\ell  \P(S_{m,\ell} = S_m)%
      = \sum_k P^k_{n,m}%
    \end{align*}
    which is the same probability of being routed from node $n$ to
    node $m$ in the original network. This is true for all $n$ and $m$
    this also implies that the probability of exiting the system after
    service at supernode $n$ is the same as the probability of exiting
    the system after leaving node $n$ in the original network.

    Now consider the relationship between service times and routing
    decisions. If a customer at supernode $n$ experiences a service
    time that is distributed according to $S_{n,k} | (S_{n,k} = S_n)$,
    that customer is then routed to node $(m, \ell)$ with probability
    $P^k_{n,m} \P(S_{m,\ell} = S_m)$. Since
    \begin{align*}
      \sum_\ell P^k_{n,m}\P(S_{m,\ell} = S_m)%
      &= P^k_{n,m}\sum_\ell \P(S_{m,\ell} = S_m)\\
      &= P^k_{n,m}
    \end{align*}
    this shows that relationship between service times and routing
    decision is maintained.
  \end{enumerate}
  
  Because the external arrivals and the routing decisions are
  equivalent, the customer flows between supernodes are equivalent to
  the corresponding customer flows between nodes in the original
  network. We can conclude that the supernodes in the $N \times K$
  node network are equivalent to the nodes in the $N$ node network.

  Because we have shown equivalence of the networks, we can find the
  stationary distribution of the original network by finding the
  stationary distribution of the new network. Each node in the new
  network is a $M/G/\infty$ queue and the routing between nodes is
  Markovian. As a result, we can apply Theorem~\ref{thrm:walrand}. Let
  $X_{n,k}$ be the number of customers at node $(n,k)$ in
  equilibrium. Equation~\ref{eq:flow} defines the flows in the
  network: $\alpha_{n, k}$ is the total arrival rate to node $(n,
  k)$. Indeed, $\P(S_{n,k} = S_n)\lambda_n$ is the external arrival
  rate to node $(n, k)$. In addition, customers are routed from node
  $(m, \ell)$ to node $(n, k)$ with probability $\P(S_{n,k} =
  S_n)P^\ell_{m,n}$ so the second term in Equation~\ref{eq:flow} is
  the internal arrival rate. We are assuming that the network is open
  so we know that Equation~\ref{eq:flow} has a unique solution that is
  non-negative~\footnote{See \cite[Chapter~2]{Chen_Yao} for details.}.
  Since $\mu_{n,k}$ is the service rate at node $(n, k)$ and node $(n,
  k)$ is an infinite server queue, we have that
  $$X_{n,k} \sim \Poisson(\rho_{n,k})$$
  where $\rho_{n,k} = \alpha_{n,k}/ \mu_{n,k}$. Since
  $$X_n \overset{D}{=} X_{n,1} + \cdots + X_{n, K}$$
  we have that $X_n \sim \Poisson(\rho_n)$. The stochastic equivalence
  of the equilibrium distribution of the entire network allows us to
  conclude that the desired product-form distribution holds.
\end{proof}

\section{Simulation Verification\label{sec:simulation}}
In this section we focus on a simple example. We illustrate the proof
of Theorem~\ref{thrm:result} by explicitly showing how the constructed
network corresponds to the original network. We then simulated the
original network to demonstrate that the product-form in
Theorem~\ref{thrm:result} is correct. We also compare these
simulations results with the two na\"ive approaches from
Section~\ref{sec:jackson} (assuming full insensitivity and assuming
insensitivity of the service times) to show that the na\"ive approaches yield
incorrect answers.

\begin{figure}
  \begin{subfigure}{\columnwidth}
    \vspace{2mm}
    \centering
    \includegraphics[width=\textwidth]{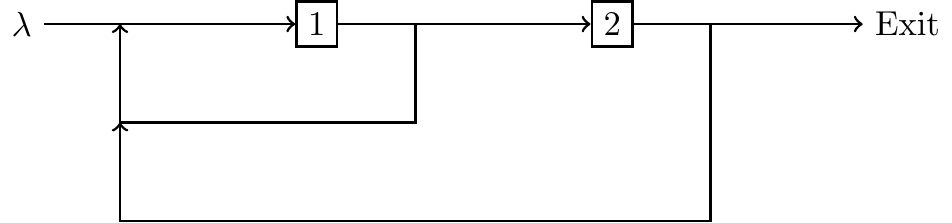}
    \caption{\label{fig:network_orig}Original Network}
    \vspace{2mm}
  \end{subfigure}
  \begin{subfigure}{\columnwidth}
    \centering
    \includegraphics[width=\textwidth]{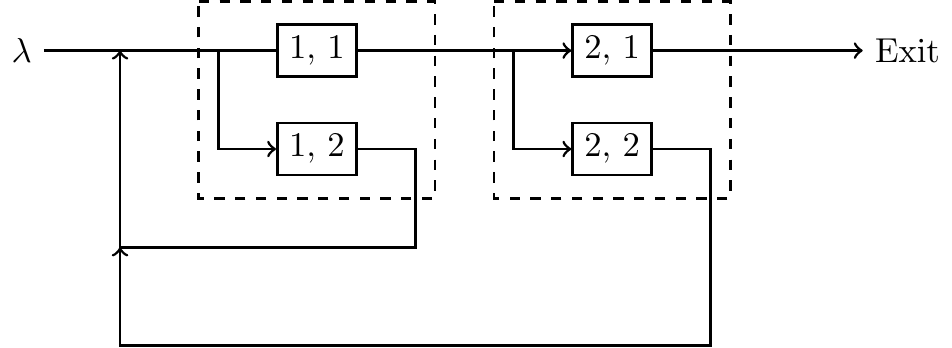}
    \caption{\label{fig:network_equiv}Equivalent Network}
  \end{subfigure}
  \caption{The network and the equivalent network discussed in
    Section~\ref{sec:simulation}.  The supernodes in the equivalent
    network are boxed with dashed lines.\label{fig:network}}
\end{figure}

We consider a the two node network in
Figure~\ref{fig:network_orig}. This network is simple model of
customers arriving to an online web service to complete two tasks in
sequence (e.g. a log-in process followed by a credit card
transaction). If customers do not complete either task within a fixed
deadline, then the customer is must start again at the
beginning. Specifically, the network can be described as follows:
\begin{itemize}
\item New customers arrive at node 1 according to a Poisson process of
  rate $\lambda$.
\item At node 1, customers have a service time $S_{1,1}$ that is
  exponentially distributed with rate $\mu_1$. There is also a
  deterministic deadline of $S_{1,2}$.
\item If $S_{1,1} < S_{1,2}$, then service is completed and the
  customer is routed with probability 1 to node 2. If $S_{1, 1} >
  S_{1, 2}$ then the service time exceeded the deadline and the
  customer is routed back to node 1.
\item At node 2, customers have a service time $S_{2, 1}$ that is
  exponentially distributed with rate $\mu_2$ and a deterministic
  deadline of $S_{2,2}$.
\item If $S_{2,1} < S_{2,2}$ then the customer exits the system and if
  $S_{2,1} > S_{2,2}$ then the service time was exceed and the
  customer is routed back to node 1.
\end{itemize}

The equivalent network is shown in
Figure~\ref{fig:network_equiv}. This network can be described as
follows:
\begin{itemize}
\item New customer arrivals are routed to node $(1, 1)$ with
  probability $1 - \exp(-\mu_1S_{1,2})$ and to node $(1, 2)$ with
  probability $\exp(-\mu_1S_{1,2})$. 
\item At node $(1,1)$ the service time is $S_{1,1} | (S_{1,1} <
  S_{1,2})$. Since this is a truncated exponential, the service rate
  at $(1, 1)$ is
  $$\mu_{1,1} = \left(\frac{1 - (\mu_1 S_{1,2} + 1)\exp(-\mu_1 S_{1,2})}{\mu_1 (1 - \exp(-\mu_1 S_{1,2}))}\right)^{-1}.$$
\item After service at node $(1,1)$, customers are routed to node $(2,
  1)$ with probability $1 - \exp(-\mu_2 S_{2,2})$ and to node $(2, 2)$
  with probability $\exp(-\mu_2 S_{2,2})$. 
\item The service rate at node $(1, 2)$ is $S_{1,2}$.
\item After service at node $(1, 2)$, customers are routed to node
  $(1, 1)$ with probability $1 - \exp(-\mu_1S_{1,2})$ and to node $(1,
  2)$ with probability $\exp(-\mu_1S_{1,2})$.
\item At node $(2,1)$ the service time is $S_{2,1} | (S_{2,1} <
  S_{2,2})$. Since this is a truncated exponential, the service rate
  at $(2, 1)$ is
  $$\mu_{2,1} = \left(\frac{1 - (\mu_2 S_{2,2} + 1)\exp(-\mu_2 S_{2,2})}{\mu_2 (1 - \exp(-\mu_2 S_{2,2}))}\right)^{-1}.$$
\item After service at node $(2, 1)$, customers exit the system.
\item At node $(2, 2)$ the service time is $S_{2,2}$.
\item After service at node $(2, 2)$, customers are routed to node
  $(1, 1)$ with probability $1 - \exp(-\mu_1S_{1,2})$ and to node $(1,
  2)$ with probability $\exp(-\mu_1S_{1,2})$.
\end{itemize}

Given this information, we can simulate the original network and
compare the empirical distribution to the distribution from
Theorem~\ref{thrm:result}. In addition, we can compare the
distribution attained by assuming full insensitivity and the
distribution attained by assuming insensitivity of the service
times. For simplicity, we focus on the case of $\lambda = S_{1,2} =
S_{2,1} = \mu_1 = \mu_2 = 1$. We simulate for $T = 10^4$ time units
with a time discretization of $dt = 10^{-3}$. We take a single run of
the simulation and take time-averages to estimate the true
distributions. We note that by relying on a single simulation run, we
are using the fact that the system is ergodic.

\begin{table}[b]
  \centering
  \begin{tabular}{|l||c|c|}
    \hline
    {} & Node 1 & Node 2\\
    \hline
    Simulated & 1.591 & 1.003\\
    \hline
    Exact & 1.582 & 1.000\\
    \hline
    Assuming Full Insensitivity & 2.000 & 1.000\\
    \hline
    Assuming Insensitivity of the Service Times & 1.251 & 0.791\\
    \hline
  \end{tabular}
  \caption{A comparison of the average number of customers at each node estimated via simulation and computed three different ways. The exact method refers to the result from Theorem~\ref{thrm:result} while the other two methods are na\"ive solution attempts discussed in Section~\ref{sec:jackson}.\label{tab:sim}}
\end{table}

First we compare the average number of customers in each
node. Table~\ref{tab:sim} shows the results. We see that although
assuming full insensitivity gives an accurate numerical result for
node 2, the na\"ive methods are both wildly incorrect for node 1. In
contrast, the exact result that is computed using
Theorem~\ref{thrm:result} agrees with the simulation.

Now that the na\"ive methods are seen to be inadequate, we can now
verify that the exact distribution agrees with the simulated
distribution. First we consider the marginal distributions of
customers at node 1 and at node 2. The results are shown in
Figure~\ref{fig:marginals}. As expected, the exact result agrees with
the simulation.

\begin{figure}
\begin{subfigure}{\columnwidth}
    \centering
    \includegraphics[width=\textwidth]{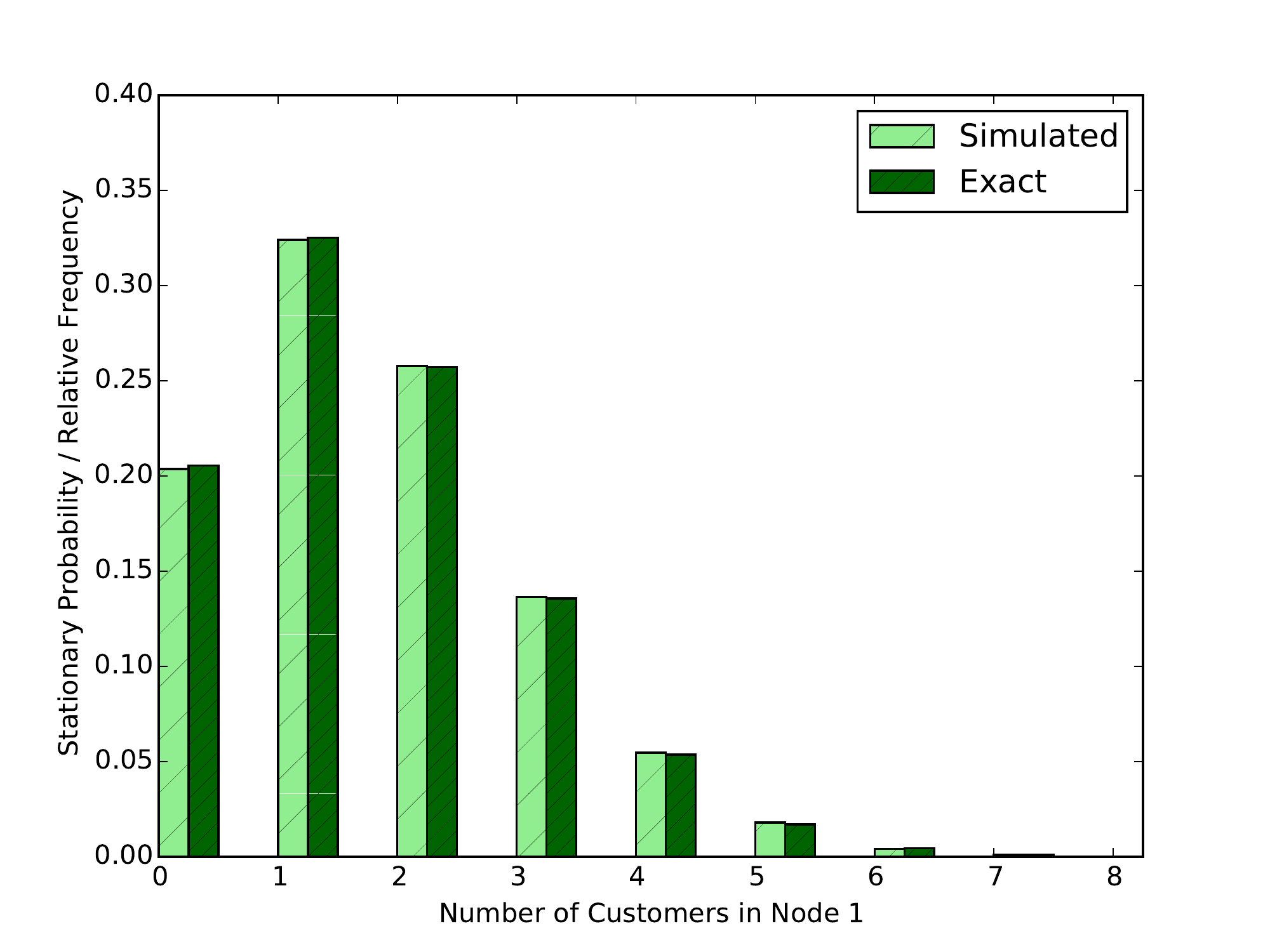}
    \caption{\label{fig:node1}Marginal Distribution of Node 1}
  \end{subfigure}
  \begin{subfigure}{\columnwidth}
    \centering
    \includegraphics[width=\textwidth]{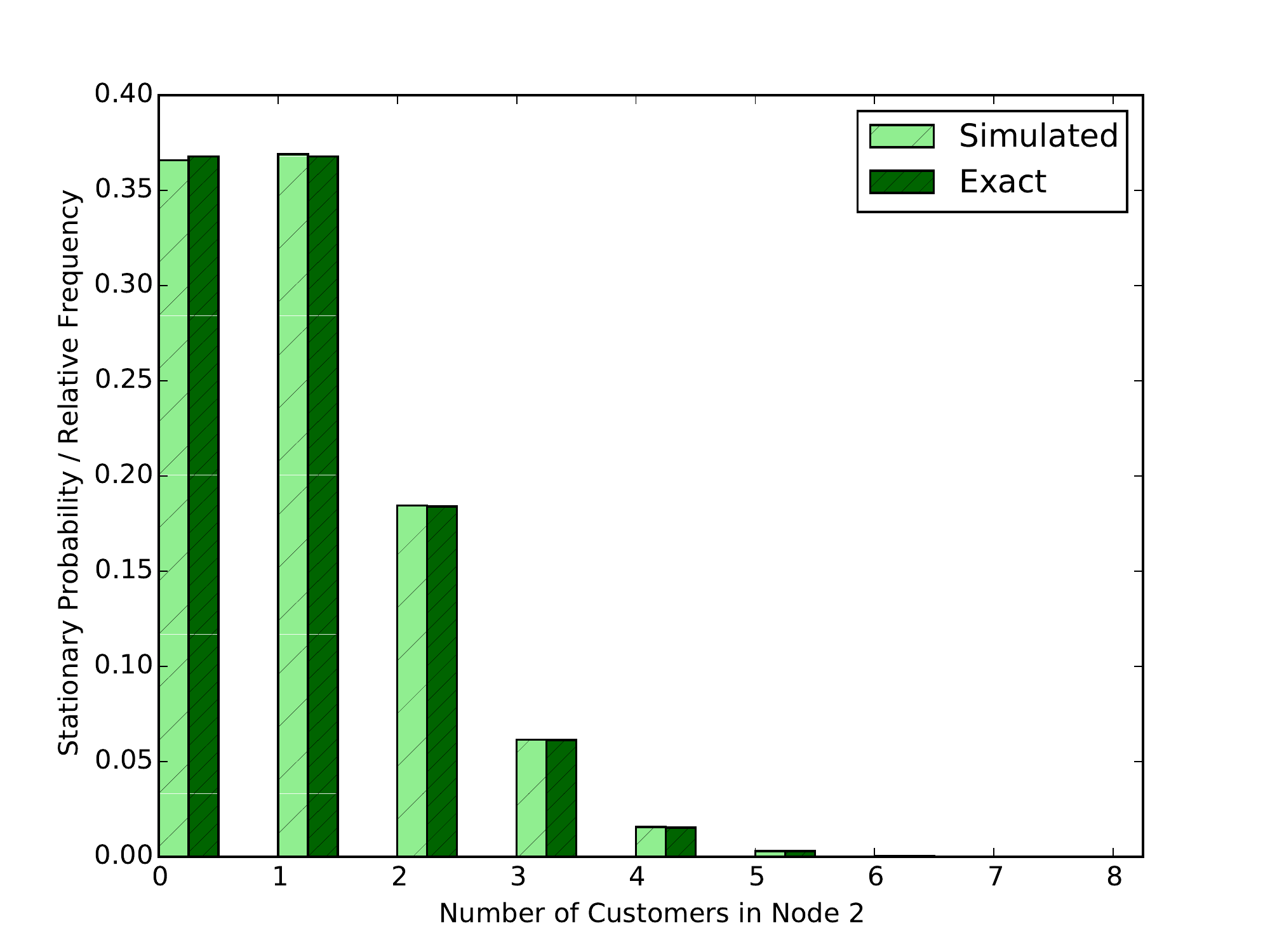}
    \caption{\label{fig:node2}Marginal Distribution of Node 2}
  \end{subfigure}  
  \caption{A comparison of the simulated marginal distributions and
    the exact marginals at each node. We see that the exact result
    from Theorem~\ref{thrm:result} agrees with the
    simulation.\label{fig:marginals}}
\end{figure}

Theorem~\ref{thrm:result} also says that the distribution is the
product of the marginals. Let $\hat\pi(x_1, x_2)$ be the empirical
probability of having $x_1$ customers at node 1 and $x_2$ customers at
node 2. Let $\hat\pi_1(x_1)$ be the empirical probability of having
$x_1$ customers at node 1. Let $\hat\pi_2(x_2)$ be the empirical
probability of having $x_2$ customers at node 1. Note that
$\hat\pi_1(\cdot)$ and $\hat\pi_2(\cdot)$ are shown in
Figure~\ref{fig:marginals} and they agree with the result from
Theorem~\ref{thrm:result}. In our simulation
$$\sup_{(x_1, x_2) \in \Zbb_+^2} \left| \hat\pi(x_1, x_2) - \hat\pi_1(x_1)\hat\pi_2(x_2)\right| = 0.00219$$
when rounded to three significant figures. This shows that the
empirical equilibrium distribution of the network is (approximately) a
product form distribution. This agrees with the product-form result
from Theorem~\ref{thrm:result}.

\section{Conclusions and Future Work\label{sec:conclusion}}
Motivated by timeouts in Internet services, we have formulated a model
for infinite server queueing networks in which routing decisions are
based on service time deadlines. In spite of the fact that the usual
theory does not apply, we have shown that the model does indeed have a
product-form stationary distribution. In addition to providing an
analytic proof, we have also provided a simulation which verifies the
results.

We have already noted that our proof is heavily dependent on the fact
that each node in the network has an infinite number of servers. As a
result, our analysis will not immediately transfer to more general
networks in which the nodes have finitely many servers. However, our
analysis does extend to the case in which only a portion of the
overall network has infinite server queues with deadline based
routing. We may also be able to extend our analysis to the case of
closed networks of infinite server queues. 

A less straightforward and more mathematically demanding extension of
these results would be to the case of more general arrival
processes. Nonstationary Poisson arrivals to infinite server queueing
networks were considered in \cite{Massey_1993} with the main result
being that product-form results still hold but are time-varying. Given
these previous results, this seems like a fruitful direction for future
work.

Finally, we note that deadlines are essential to many applications
besides Internet services such as wireless communication
\cite{Master_ICC2014, Master_ACC2015}, patient scheduling
\cite{Master_Myopic}, low latency computing \cite{LL1,LL2}, and
utility computing \cite{Z1,Z2,Z3}. Consequently, we feel that we may
be able to apply similar modeling ideas to other application domains.

\bibliographystyle{ieeetr}
\bibliography{NMaster_NBambos_CDC2016}

\end{document}